\documentclass[letterpaper, 10pt, conference]{ieeeconf}
\IEEEoverridecommandlockouts
\usepackage{amsmath} 
\usepackage{amssymb}  
\usepackage{graphicx}
\usepackage{amsfonts}
\usepackage{float}
\usepackage{color}
\usepackage{xcolor}
\usepackage{subcaption}
\usepackage[font=footnotesize]{caption}
\usepackage{array}
\usepackage[markup=default]{changes}
\usepackage{booktabs,array,dcolumn}
\newcolumntype{P}[1]{>{\centering\arraybackslash}p{#1}}
\usepackage{mathtools}
\usepackage{algpseudocode}
\usepackage{listings}
\usepackage{cite}
\usepackage{color}
\usepackage{comment}
\usepackage[ruled,vlined]{algorithm2e}
\usepackage{graphicx}
\usepackage{color}
\usepackage{adjustbox}
\usepackage{centernot}
\usepackage{csquotes}
\usepackage{mathtools}
\usepackage{mathabx}
\usepackage{physics}
\usepackage{float}
\usepackage{multirow}
\usepackage{enumerate}
\usepackage{color}
\usepackage{cite}

\usepackage{xcolor}
\usepackage[colorlinks,citecolor=cyan,urlcolor=blue]{hyperref}
\usepackage{hyperref}
\usepackage{etoolbox}
\newcommand{\rom}[1]{\uppercase\expandafter{\romannumeral #1\relax}}
\renewcommand{\arraystretch}{1.4}
\usepackage{bm}
\usepackage{balance}
\usepackage{mathtools}

\newtheorem{theorem}{Theorem}
\newtheorem{corollary}{Corollary}
\newtheorem{proposition}{Proposition}
\newtheorem{example}{Example}
\newtheorem{remark}{Remark}
\newtheorem{definition}{Definition}

\usepackage{color}
\definecolor{gray}{RGB}{128,128,128}

\newcommand{\remove}[1]{}

\newcommand{\Mstar}{M^{*}}
\newcommand{\ME}{M_{\mathcal{E}}}
\newcommand{\MR}{M_{\mathcal{R}}}
\newcommand{\EE}{\mathbb{E}}

\begin{document}
\title{Finite-Sample Limits of Entropy-Based Structure Identification in Discretized Nonlinear Systems}

\author{Pratishtha Shukla and James Nutaro \\
\small{Computational Sciences and Engineering Division, Oak Ridge National Laboratory}
\thanks{*This manuscript has been authored by UT-Battelle, LLC, under contract DE-AC05-00OR22725 with the U.S. Department of Energy (DOE). The U.S. government retains and the publisher, by accepting the article for publication, acknowledges that the U.S. government retains a nonexclusive, paid-up, irrevocable, worldwide license to publish or reproduce the published form of this manuscript, or allow others to do so, for U.S. government purposes. DOE will provide public access to these results of federally sponsored research in accordance with the DOE Public Access Plan (https://www.energy.gov/doe-public-access-plan).}
\thanks{Computational Sciences and Engineering Division, Oak Ridge National Laboratory, US
        {\tt \small shuklap@ornl.gov, nutarojj@ornl.gov}}%
}
\maketitle

\begin{abstract}
Discretization fundamentally limits structure identification in stochastic systems. When system stochasticity exceeds the discretization resolution, entropy-based methods lose their ability to distinguish which input drives the output. We study this in Fuzzy Inductive Reasoning (FIR), a nonparametric framework for learning dynamical systems from discretized measurements, where the choice of input variables determines both predictive accuracy and the interpretability of the learned input--output relationships. Entropy-based selection targets explainability, i.e., identifying which variables causally drive the output, while mean-squared-error-based selection targets prediction. We introduce a resolution-stochasticity ratio that governs when entropy-based selection is reliable. Three results follow. First, entropy-based selection is consistent below this threshold but loses discriminative power above it, regardless of sample size. Second, using the entropy-selected variables for prediction instead of the MSE-selected ones incurs a closed-form excess prediction risk that grows with input complexity and shrinks with sample size. Third, reliable identification of the causally relevant inputs requires data that scales with the number of input combinations and inversely with the strength of the entropy signal. The theory is validated on a two-state Markov model and demonstrated on a distribution grid reliability dataset analyzing the impact of infrastructure investment, where the goal is to explain which investments drive reliability improvements rather than merely predict outcomes.
\end{abstract}

\section{Introduction}
\label{sec:intro}

Fuzzy Inductive Reasoning (FIR)~\cite{cellier1996combined} is a framework to learn input-output dynamics from data. FIR maps variables to qualitative levels (e.g., ``low,'' ``medium,'' ``high'') and constructs probabilistic transition rules from discrete patterns. FIR has been applied to systems with partially unknown dynamics, such as autonomous space systems and complex energy demand processes~\cite{tang2024fuzzy, mamlook2009fuzzy}. Similar systems arise in electricity markets, where price and demand dynamics emerge from interactions among multiple agents and operational constraints, motivating the use of data-driven forecasting methods~\cite{kalhori2022data}.

A central design decision in FIR is to select the input variables; specifically, which variables and their lag in time~\cite{nebot2012fuzzy}. This selection is called a mask. The choice of input variables serves two distinct goals: \emph{explainability}, i.e., identifying which variables causally drive the output so that the learned relationships make domain sense, and \emph{prediction}, i.e., minimizing forecast error. Entropy-based selection targets explainability by minimizing residual uncertainty given the inputs; mean-squared-error (MSE) based selection targets prediction accuracy~\cite{acosta2007optimization}. Too few variables may omit relevant structure, while too many lead to sparsely populated rule tables and degraded predictions.

Despite its importance, the choice of a metric for mask selection is essentially empirical~\cite{bagherpour2015hierarchical,jurado2017fuzzy}. In part, this is because we lack a precise understanding of how the metric for ranking masks, uncertainty in measurements of the system, and discretization resolution jointly affect the quality of the mask selected. Specifically, it remains unclear when an entropy-based metric reliably recovers the underlying input structure and when it systematically favors overly complex masks in sparse, noisy settings.

Information-theoretic feature selection methods based on mutual information (MI), such as mRMR~\cite{peng2005feature, papaioannou2025role} and conditional MI maximization~\cite{covert2023learning}, provide asymptotic consistency in continuous settings but do not account for discretization ceilings or finite-sample sparsity. A recent review of information-theoretic variable selection~\cite{mielniczuk2022information} notes that empirical conditional MI performs poorly under sparse discrete distributions, precisely the setting studied here. Heuristic FIR mask selection via genetic algorithms~\cite{acosta2007optimization} lacks finite-sample guarantees. To our knowledge, no prior work derives a closed-form excess prediction risk for discretized models, nor jointly characterizes the effect of stochasticity and discretization resolution on structural recovery.

To fill this gap, we develop a theoretical analysis of mask selection under entropy- and mean-squared-error-based criteria. We introduce a resolution-stochasticity quantity that captures the interaction between intrinsic variability and discretization, and show that it fundamentally limits the discriminative power of entropy-based methods. We characterize regimes where entropy-based selection is reliable and where it becomes unreliable, while mean-squared-error-based objectives remain robust for prediction. 

This framework provides a practical guideline for choosing between entropy and MSE-based methods and it quantifies both the excess prediction risk arising from entropy-driven over-selection and the sample requirements for accurate mask recovery. Although developed for FIR, the results characterize a general information limit in discretized nonparametric models with sparse rule tables. Furthermore, the grid reliability application demonstrates the practical value of the method to identify which infrastructure investments drive reliability improvements.



\section{Fuzzy Inductive Model}
\label{sec:model}
FIR models dynamical systems using discretized input–output relationships learned from data. Conceptually, it partitions the input space into discrete values and estimates a conditional output distribution for each observation of these values. The model identification problem is to select a set of time lagged input that best explains the observed dynamics. 


\subsection{System and Mask}
We consider a discrete-time system with a single output variable and one or more input variables ($1, 2 \dots, d$) that are observed over time. The variables are discretized to take values $1, 2, \dots, q_i$, where $q_i$ denotes the number of discrete values for variable $i$. We assume that the observed sequence has length $L$. The discrete value of input variable $d$ at time $t$ is denoted $X_t^d$ and for the output $Y_t$. We build a model of this system by selecting a particular \emph{mask} of the input sequence and mapping observations of that mask to an anticipated output. The discovery of a best \emph{mask} is the goal of the FIR procedure.

\begin{example}
Consider two binary variables observed over time, with the goal of predicting the next value of the lower row given past values of both:
\begin{align*}
\text{Upper:}\quad
&1\ 1\ 0\ 1\ 0\ 0\ 1\ 0\ 0\ 0\
  0\ 0\ 1\ 0\ 0\ 0\ 1\ 0\ 1\ 0\\
\text{Lower:}\quad
&1\ 0\ 1\ 1\ 0\ 0\ 0\ 1\ 1\ 1\
  1\ 1\ 1\ 0\ 0\ 0\ 0\ 1\ 1\ 0
\end{align*}
A \emph{mask} selects which past values to use as predictors. For a first attempt, \textbf{Mask 1} uses only the previous upper value. This generates twenty input to output pairs, some of which are repeated. 
\begin{center}
\small
\begin{tabular}{ccc}
\hline
\textbf{Input} &
\textbf{Output} &
\textbf{Count} \\
\hline
0 & 0 & 5 \\
0 & 1 & 8 \\
1 & 0 & 4 \\
1 & 1 & 3 \\
\hline
\end{tabular}
\end{center}
A better model, i.e., \textbf{Mask 2} uses both previous values of upper and lower rows.
\begin{center}
\small
\begin{tabular}{ccc}
\hline
\textbf{Input} &
\textbf{Output} &
\textbf{Count} \\
\hline
(0,0) & 0 & 5 \\
(0,1) & 1 & 7 \\
(1,0) & 1 & 3 \\
(1,1) & 1 & 4 \\
\hline
\end{tabular}
\end{center}
Mask 2 produces more informative predictions, each input pattern maps to a near-certain output. The choice of mask determines model quality. 
\end{example}

\begin{definition}[Mask] \label{def:mask}
A mask is a subset
\begin{align}
M \subseteq \{1,\dots,d\} \times
\{0,\dots,L\},
\end{align}
whose elements $(i,\ell)$ specify which variable--lag pairs are used as input by the mask. For a mask $M$, the corresponding input symbol at time $t$ is a tuple $X_t^M$ that contains the elements $X_{t-\ell}^i$ for each $(i,\ell)$ in $M$.
\end{definition}

The output at time $t$ given a mask $M$ is described by a conditional probability distribution
\begin{equation}
Y_t \mid X_t^M \sim P(Y_t \mid X_t^M),
\label{eq:system}
\end{equation}
with condition mean as
\begin{equation}
\mu(X_t^M) =
\mathbb{E}[Y_t \mid X_t^M].
\label{eq:condmean}
\end{equation}
Real data inherently has variability. If this is not fully captured by the selected mask, the system is said to be stochastic, with residual variability $\sigma_\varepsilon^2 = \mathrm{Var}(Y_t \mid X_t^M)$.

\begin{example}
With $d=2$ variables and $L=2$, there are $2\times3=6$ possible masks. The mask $M=\{(1,2),(2,1)\}$ selects variable~1 at lag~2 and variable~2 at lag~1, giving $X_t^M = (X_{t-2}^1,X_{t-1}^2)$.
\end{example}

\subsection{Input Cells and the Rule Table}
Each component of $X_t^M$ takes a discrete value, so $X_t^M$ takes values in a finite set of possible combinations. Each such combination defines a cell. There are
\begin{equation} \label{eq:Km}
    K(M) = \prod_{(i,\ell)\in M} q_i
\end{equation}
such combinations under mask $M$. Each unique combination
defines an \emph{input cell} $B_k$,
for $k = 1,\ldots,K(M)$. Note that $B_k$ is a region of the discretized input space, not an
observation. For a fixed mask
$M$, there are exactly $K(M)$ cells
regardless of how many observations
the data contains.

\begin{example}
With $d=2$ variables, $q_1=q_2=3$ bins each, and mask $M=\{(1,1),(2,1)\}$, we have $K(M)=9$ input cells. Cell $B_1 =(0,0)$ where both input variables take value zero at lag~1. The remaining cells are $B_2 = (0,1)$, $B_3 = (0,2)$, $B_4 = (1,0)$, $B_5 = (1,1)$, $B_6 = (1,2)$,  $B_7 = (2,0)$, $B_8 = (2,1)$, and $B_9 = (2,2)$.
\end{example}

The cell $B_k$ appears some number of times in our set of observations. Each such appearance is associated with a value $y$ that can be taken by $Y_t$. Let $n(B_k,y)$ be the number of instances where $y$ is observed in response to $B_k$ and $n(B_k)$ the total number of instances of $B_k$ in the data. The FIR model estimates the conditional output distribution empirically by
\begin{equation} \label{eq:count}
  \hat{P}(Y{=}y \mid B_k) = \frac{n(B_k,\, y)}{n(B_k)} = \hat{p}_{yk}
\end{equation}
The collection of conditional distributions across all cells is the \emph{rule table}, i.e., the complete FIR model.  If the possible values of the output are $y_1, \dots, y_r$ then the sample average of output of cell $B_k$ is 
\begin{align} \label{eq:out-est-mu}
\hat{\mu}_{B_k} =
\sum_{j=1}^{r}
y_j \cdot \hat{p}_{jk}
\end{align}
To simplify our analysis, we assume that each cell $B_k$ appears at least once in the data set.


\subsection{Mask Selection Objectives}
Two objectives are used to select among candidate masks:
\begin{align}
  \ME &= \arg\min_{M}\;
        \hat{H}(M),
        \label{eq:entropy_obj}\\
  \MR &= \arg\min_{M}\;
        \hat{R}^2(M),
        \label{eq:rmse_obj}
\end{align}
Suppose that our empirical data set has $N_s$ observations of cells paired with an output. Let $y_i$ be the output of the $i$th observation and $B_i$ the cell to which it belongs. The 
empirical conditional entropy of mask $M$ is
\begin{equation} \label{eq:entropy}
 \hat{H}(M) = \sum_{k=1}^{K(M)}
\hat{P}(B_k)\, \hat{H}(B_k).
\end{equation}
The term $\hat{H}(B_k)$ is the entropy estimator for cell $k$ given by
\begin{equation*}
   \hat{H}(B_k) = -\sum_j
\hat{p}_{jk} \log_2 \hat{p}_{jk},
\end{equation*}
with $\hat{p}_{jk}$ is the probability of cell $B_k$ producing output $j$ as obtained from empirical cell counts (\ref{eq:count}); the sum is over all possible values of $j$. For the same data, the empirical MSE relative to a data set is given by
\begin{align} \label{eq:MSE}
  \hat{R}^2(M) = \frac{1}{N_s}
  \sum_{i=1}^{N_s}
\bigl(y_i -
\hat{\mu}_{B_i}\bigr)^2,
\end{align}
where $\hat{\mu}_{B_i}$  (\ref{eq:out-est-mu}) is the estimated mean of the cell containing observation $i$.

The MSE (\ref{eq:MSE}) and entropy (\ref{eq:entropy}) are distinct. Entropy $\hat{H}$ depends on the full conditional distribution while $\hat{R}^2$ depends only on the conditional mean. Entropy-based mask selection relies on in-sample estimates of coarsened conditional entropy, whereas model quality is ultimately measured by out-of-sample prediction error. The gap between these objectives underlies the behavior analyzed in this paper.

\subsection{Sparsity Parameter}
As the mask size increases, the number of cells grows exponentially with $|M|$, causing many cells to contain few or no observations when data are limited. Estimating a distribution requires more data per cell than estimating a mean. This asymmetry becomes critical in the sparse regime. We quantify sparsity by
\begin{align} \label{eq:lambda}
  \lambda = N_s/K(M).
\end{align} 
Since $\lambda$ varies across candidate masks, sparsity regimes are characterized using the worst-case cell count $K(M_{\max})$, where $M_{\max} = \arg\max_{M} K(M)$. Small values of $\lambda$ correspond to sparse rule tables, while large values indicate well-populated cells.

\subsection{Resolution-Stochasticity Ratio}

The output space is partitioned into $q$ equal-sized intervals. In the normalized output space $[0,1]$, each bin has width $\Delta = 1/q$. Now suppose the system output depends on a finite history of past inputs through the conditional mean $\mu(X^M_t)$ (eq.~\ref{eq:condmean}), but exhibits intrinsic stochasticity that follows $\mathcal{N}(0,\sigma_\varepsilon^2)$.


The intrinsic variability $\varepsilon$ spreads the measured output by $\sigma_\varepsilon$ around its expected value. The ratio
\begin{equation}
  \rho = \frac{\sigma_\varepsilon}{\Delta} =
  \sigma_\varepsilon \cdot q
  \label{eq:rho}
\end{equation}
measures how many discrete output values the variability spans. It therefore compares the stochastic variability of the measured system with the resolution of the discretized representation. 

Assuming $\mu(X_t^M) \approx \mu_k$ for all observations in $B_k$, the conditional distribution of
$Y_t$ in cell $B_k$ is $\mathcal{N}(\mu_k, \sigma_\varepsilon^2)$ and the probability of output bin $j$ is:
\begin{align}
p_{jk} = \Phi\!\left(
\tfrac{(j{+}1)/q - \mu_k}
{\sigma_\varepsilon}\right) -
\Phi\!\left(\tfrac{j/q - \mu_k}
{\sigma_\varepsilon}\right),
\label{eq:CDF}
\end{align}
where $\Phi$ is the standard normal CDF. These assumptions support the theoretical analysis of Section~\ref{sec:theory}; the rule table estimates $\hat{P}(Y \mid B_k)$ directly from cell counts without any distributional assumption on the intrinsic stochasticity.

Fig.~\ref{fig:fir_cells} illustrates the role of $\rho$. When $\rho < 1$, each cell produces a distinct output distribution. When $\rho > 1$, all cells produce approximately uniform distributions $p_{jk} \approx 1/q$.

\begin{figure}[t]
\centering
\includegraphics[width=\columnwidth]
{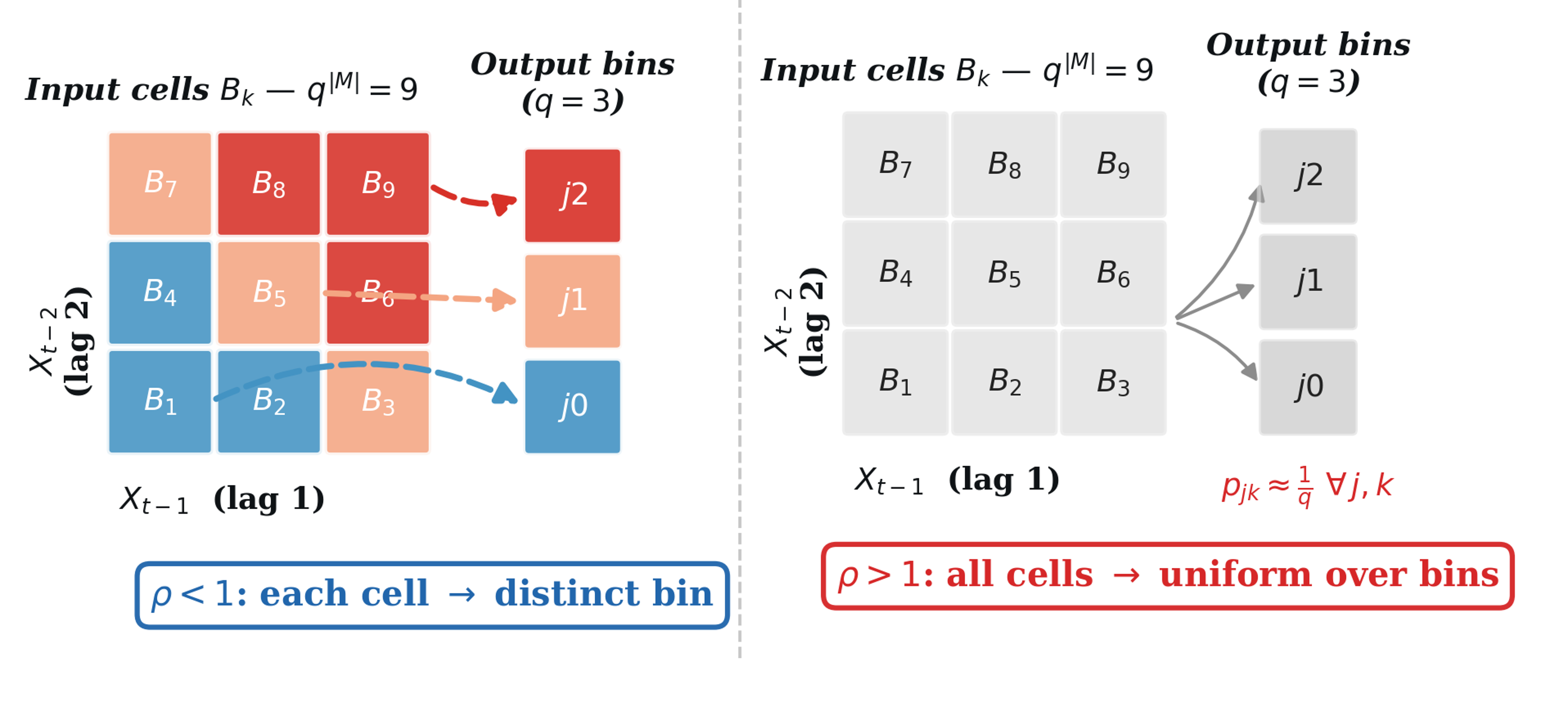}
\caption{FIR rule table structure for a
two-input mask ($q_v=3$, $|M|=2$,
$K(M)=9$ cells). Color indicates the
dominant output bin.}
\label{fig:fir_cells}
\end{figure}

\section{Theoretical results}
\label{sec:theory}

\label{subsec:notation}
\begin{definition}[Minimal Sufficient Mask]
The mask $M^\star$ is \emph{minimal sufficient} if, for all $t$, $Y_t$ is conditionally independent of all inputs not in $X_t^{M^\star}$, and no proper subset of $M^\star$ satisfies this condition. Hence $M^\star$ contains exactly the variable-lag pairs that carry all predictive information about $Y_t$ with no redundancy.
\end{definition}

Recovering $M^\star$ from data is the goal of mask selection. The primary information quantity for this purpose is the \emph{coarsened conditional entropy} $H_q(M)$. This is defined as the conditional entropy of the discretized output $Y_t$, taking $q$ discrete values, given the inputs selected by mask $M$. This differs from the continuous entropy $H(\varepsilon) = \tfrac{1}{2}\log(2\pi e \sigma_\varepsilon^2)$, which grows unbounded with $\sigma_\varepsilon$. Because the discretized output takes at most $q$ values, $H_q(M)$ is bounded:
\begin{align}
H_q(M) \leq \log_2 q,
\label{eq:hq_ceiling}
\end{align}
with equality only when the $q$ output bins are equiprobable~\cite{cover2006elements}. This \emph{discretization ceiling} is a hard constraint independent of the underlying stochasticity. We define the intrinsic coarsened entropy as $H_q(\varepsilon) = H_q(M^\star)$, i.e., the coarsened
conditional entropy under the minimal sufficient mask. The $q$ subscript distinguishes this from
the unbounded continuous entropy $H(\varepsilon)$. 

\subsection{Main Results}

The following three results characterize mask selection. 

\subsubsection{Resolution Regime and Consistency}
Entropy masking is said to be \emph{consistent} if the selected mask contains $M^\star$ with high probability as the number of samples $N_s$ increases. Exact recovery corresponds to selecting $M^\star$ itself. MSE masking is \emph{prediction-consistent} if its prediction error converges to the irreducible variability level. Entropy-based mask selection relies on distinguishing conditional distributions across input cells. The following result identifies when this distinction is possible. 

\begin{theorem} \label{thm:threshold}
Consider input--output observations $(X_t, Y_t)$ generated with intrinsic variability as $\varepsilon_t \sim
\mathcal{N}(0,\sigma_\varepsilon^2)$, and $q$ output bins. Let $\rho = \sigma_\varepsilon q$.
\begin{enumerate}
    \item[\textbf{(i)}] \textbf{Resolved regime ($\rho < 1$):} Assume $\lambda>1$. 
    The coarsened conditional entropy under the minimal sufficient mask satisfies
    \begin{align*}
        H_q(\Mstar) < \log_2 q.
    \end{align*}
    Moreover, as $N_s \to \infty$, entropy masking is consistent: $\mathbb{P}(\ME \supseteq \Mstar) \to 1.$

\item[\textbf{(ii)}] \textbf{Saturation regime ($\rho > 1$):}
    The coarsened conditional entropy reaches the ceiling for all masks:
    \begin{align*}
        H_q(M) \to \log_2 q
        \qquad \forall M,
    \end{align*}
    so entropy masking is inconsistent for any $N_s$.
    
    \item[\textbf{(iii)}]\textbf{Ceiling regime ($\rho \approx 1$):}
        The coarsened conditional entropies of all masks, including $M^\star$, approach the discretization ceiling:
        \begin{align*}
        H_q(M) \approx \log_2 q.
        \end{align*}
    Entropy-based discrimination between masks degrades.

    \item[\textbf{(iv)}] \textbf{MSE prediction-consistency (all $\rho$):}
    For any mask $M \supseteq \Mstar$, $\hat{R}^2(M) \to \sigma_\varepsilon^2 \quad \text{as } N_s \to \infty.$  MSE masking is prediction-consistent for all $\rho$, although it does not uniquely recover $\Mstar$.
\end{enumerate}
\end{theorem}
\begin{proof}
Fix mask $M$ and input cell $B_k$ with conditional mean $\mu_k$. Under eq.~\eqref{eq:system}, $Y_t \mid B_k \sim \mathcal{N}(\mu_k, \sigma^2_\varepsilon)$. The conditional cell-wise output probability is $p_{jk}$ (eq.~\eqref{eq:CDF}).

\noindent\textbf{(i) Resolved regime ($\rho < 1$):} When $\rho < 1$, the intrinsic variability scale is smaller than the output bin width $\Delta = 1/q$. For each input cell $B_k$, most conditional probability mass concentrates in the output bin containing $\mu_k$. Denoting this bin $j^\star(k)$, we have $p_{j^\star(k),k} \approx 1$, and consequently
\begin{align*}
H_q(B_k) = -\sum_{j=1}^{q} p_{jk}\log_2 p_{jk} \approx 0.
\end{align*}
Since $H_q(M^\star)$ is the cell-wise average of $H_q(B_k)$ under $M^\star$, it follows that $H_q(M^\star) < \log_2 q$, with $H_q(M^\star) \approx 0$ in the strongly resolved limit ($\rho \to 0$) and $H_q(M^\star) = 0$ only in the deterministic case
(Proposition~\ref{prop:deterministic}).

The role of distinct conditional means across cells is not to increase $H_q(M^\star)$, but to make the relevant variables identifiable. If $M \not\supseteq M^\star$ omits a relevant variable, observations from cells with distinct $\mu_k$ are pooled into the same cell. The resulting output distribution is a mixture of otherwise distinguishable cell-wise distributions, strictly increasing conditional entropy $H_q(M) > H_q(M^\star)$. Hence $M^\star$ uniquely minimizes $H_q(M)$. Since $\hat{H}(M) \to H_q(M)$ as $N_s \to \infty$, entropy masking is consistent, i.e., $\mathbb{P}(\ME \supseteq \Mstar) \to 1.$


\noindent\textbf{(ii) Saturation regime ($\rho > 1$):} Unlike part~(i), when $\sigma_\varepsilon > 1/q$, the Gaussian variability becomes large relative to $\Delta$. Applying the first-order normal approximation $\Phi(b) - \Phi(a) \approx (b-a)\phi(m)$ to eq.~\eqref{eq:CDF}:
\begin{align*}
p_{jk} \approx
\frac{1/q}{\sigma_\varepsilon}
\phi\!\left(
\frac{j/q - \mu_k}{\sigma_\varepsilon}
\right)
\to \frac{1}{q}
\quad \text{as } \sigma_\varepsilon
\to \infty,
\end{align*}
so, in the large-variability limit, all $q$ output bins become equiprobable in every cell regardless of $\mu_k$. Therefore $H_q(M) \to \log_2 q$ for all masks $M$. The entropy objective becomes flat across masks, losing discriminative power even as $N_s \to \infty$. 

\noindent \textbf{(iii) Ceiling regime ($\rho \approx 1$):} When $\rho$ is close to $1$, entropy differences across masks become small, and discrimination degrades. The threshold $\rho_c = 1$ is a
necessary condition under Gaussian stochasticity, the boundary is not sharp, as entropy masking
degrades continuously with increasing $\rho$ rather than failing abruptly at $\rho = 1$.

\noindent\textbf{(iv) MSE prediction-consistency:}
The FIR prediction for a test observation in cell $B_k$ with output $y \in \{y_1,\ldots,y_r\}$
is the cell mean $\hat{\mu}_{B_k}$. Define the per-cell expected squared error $e_k = \EE[(y -
\hat{\mu}_{B_k})^2]$. Adding and subtracting $\mu_k$:
\begin{align} \label{eq:expectation}
e_k =
\underbrace{\EE[\varepsilon^2]}_{
\sigma^2_\varepsilon}
+
\underbrace{\EE[(\mu_k -
\hat{\mu}_{B_k})^2]}_{
\sigma^2_\varepsilon/n_k}
+
\underbrace{2\,\EE[\varepsilon
(\mu_k - \hat{\mu}_{B_k})]}_{
=\,0},
\end{align}
where the cross term vanishes because $\varepsilon$ is independent of the observations used to form $\hat{\mu}_{B_k}$. Thus $e_k = \sigma^2_\varepsilon(1 + 1/n_k) \to \sigma^2_\varepsilon$ as $n_k \to \infty$, so $\hat{R}^2(\MR) \to \sigma^2_\varepsilon$ for any $\rho > 0$. Since this holds for any $M \supseteq \Mstar$, all such masks achieve the same variability floor. Thus, $\MR$ is prediction-consistent but does not uniquely recover $\Mstar$.
\end{proof}


\subsubsection{Closed-Form Excess Risk} 
While Theorem~\ref{thm:threshold} characterizes when entropy can identify the true structure, it may still select a superset of variables. The additional variables increase model complexity without improving predictive power. The following result quantifies the prediction cost of such over-selection.

\begin{theorem}\label{thm:delta}
Assume balanced cells ($n_k = N_s/K(M)$ for all $k$) and that $\hat{R}^2(M)$ is evaluated on observations independent of those used to estimate $\hat{\mu}_{B_k}$. For mask $M$ with $K(M)$ cells:
\begin{align*}
\delta = \EE[\hat{R}^2(\ME)] -
\EE[\hat{R}^2(\Mstar)]
= \frac{\sigma^2_\varepsilon}{N_s}
\bigl(K(\ME) - K(\Mstar)\bigr).
\end{align*}
$\delta \geq 0$ is monotone increasing in over-selection $(K(\ME)-K(\Mstar))$ and $\sigma^2_\varepsilon$, and monotone decreasing in $N_s$.
\end{theorem}

\begin{proof}
Fix mask $M$ with $K(M)$ cells. Let $n_k$ denote the number of observations in cell $k$, so $\sum_{k=1}^{K(M)} n_k = N_s$. The true cell mean $\mu_k$ is assumed approximately constant within each cell, and $\hat{\mu}_{B_k}$ is its estimate from eq.~\eqref{eq:out-est-mu}.

For a test observation in $B_k$ with output $y \in \{y_1,\ldots,y_r\}$, use $e_k = \EE[(y - \hat{\mu}_{B_k})^2]$. Adding and subtracting $\mu_k$ as in eq.~\ref{eq:expectation}, 
the cross term vanishes because the test observation is independent of $\hat{\mu}_{B_k}$ by assumption.
Under balanced cells ($n_k = N_s/\lambda$), eq.~\eqref{eq:lambda}):
\begin{align*}
\frac{\sigma^2_\varepsilon}{n_k}
= \frac{\sigma^2_\varepsilon \cdot
K(M)}{N_s}
= \frac{\sigma^2_\varepsilon}{\lambda}.
\end{align*}
Averaging over all $K(M)$ cells:
\begin{align}
\EE[\hat{R}^2(M)] =
\sigma^2_\varepsilon\!\left(1 +
\frac{K(M)}{N_s}\right).
\label{eq:rmse_decomp}
\end{align}
Each additional variable $i$ multiplies $K(M)$ by $q_i$, proportionally increasing the per-cell estimation variance. The excess risk is the difference of eq.~\eqref{eq:rmse_decomp} at $\ME$ and $\Mstar$:
\begin{align} \label{eq:delta}
\delta =
\frac{\sigma^2_\varepsilon}{N_s}
(K(\ME) - K(\Mstar)).
\end{align}
Since $\ME \supsetneq \Mstar$
implies $K(\ME) > K(\Mstar)$,
$\delta \geq 0$. Monotonicity
follows directly from
eq.~\eqref{eq:delta}; monotonicity
in $H_q(\varepsilon)$ follows from
Corollary~\ref{cor:monotone}.
\end{proof}

\begin{remark}
\label{rem:jensen}
Theorem~\ref{thm:delta} assumes balanced cells. When cells are unbalanced, Jensen's inequality gives
\begin{align*}
  \delta_{\mathrm{emp}} \leq
  \frac{\sigma^2_\varepsilon}{N_s}
  \bigl(K(\ME) - K(\Mstar)\bigr),
\end{align*}
where $\delta_{\mathrm{emp}} = \hat{R}^2(\ME) -
\hat{R}^2(\MR) \geq 0$ is the empirical prediction cost of entropy
over-selection relative to the MSE-optimal mask, used as a proxy for $M^*$. When $\lambda < 1$, empty cells introduce a bias term absent from the variance decomposition and the bound does not apply. Under Poisson cell occupancy with mean $\lambda$, the fraction of empty cells is approximately $e^{-\lambda}$ and the expected per-cell variance inflation is $\lambda/(1-e^{-\lambda})$, giving a tighter bound when $\lambda > 1$.
\end{remark}

\begin{corollary}[Monotonicity in $H_q(\varepsilon)$] \label{cor:monotone}
Under the conditions of Theorem~\ref{thm:delta}, $\delta$ is monotone increasing in $H_q(\varepsilon)$: higher coarsened stochasticity entropy leads to greater over-selection by entropy masking and hence larger excess prediction risk.
\end{corollary}

\begin{proof}
From Theorem~\ref{thm:delta}, $\delta$ is given by eq.~(\ref{eq:delta}). Since $K(M^\star)$ is fixed, it suffices to show that $K(M_E)$ increases with $H_q(\varepsilon)$. As $H_q(\varepsilon)$ increases toward $\log_2 q$, the conditional output distributions induced by different masks become less distinct. Hence the entropy reduction obtained by adding a truly informative variable becomes smaller, while the finite-sample bias favoring larger masks remains. As a result, entropy masking tends to select larger masks, increasing $K(\ME)$.

Therefore $\delta$ increases with $H_q(\varepsilon)$. In the limiting case $H_q(\varepsilon)\to\log_2 q$, entropy masking approaches the largest candidate mask $M_{\max}$, and
\begin{align}
\delta \to
\frac{\sigma_\varepsilon^2}{N_s}
\bigl(K(M_{max})-K(M^\star)\bigr).
\end{align}
where $K(M_{\max})$ is the cell count for $M_{max}$
\end{proof}

\subsubsection{Sample Complexity Scaling}
While Theorem~\ref{thm:delta} quantifies the cost of over-selection, the following result characterizes how much data is required for entropy masking to reliably recover the true structure. As stochasticity increases, entropy differences between candidate masks shrink, making them harder to detect. Since the variance of entropy estimates decreases with the number of samples per cell, reliably distinguishing small entropy gaps requires more data. This leads to a sample complexity that scales inversely with the square of the effective entropy signal $H_q(\varepsilon)$.

\begin{theorem}\label{thm:sample}
Let $H_q(\varepsilon) = H_q(\Mstar)$ denote the intrinsic coarsened conditional entropy. A sufficient condition for entropy masking to satisfy $\mathbb{P}(\ME \supseteq \Mstar) \geq 1-\alpha$ is:
\begin{align}
N_s \gtrsim
\frac{C_\alpha\, K(M_{\max})}
{H_q(\varepsilon)^2},
\end{align}
where $C_\alpha = 4z_\alpha^2$ and $z_\alpha$ is the $\alpha$-quantile of the standard normal distribution.
\end{theorem}

\begin{proof}
Fix $M$ with $K(M)$ cells and $q$ output values.

\noindent\textbf{Step 1: Bias toward larger masks.}
By the Miller--Madow correction~\cite{miller1955note}, the plug-in entropy estimator has a systematic negative bias:
\begin{align*}
\mathbb{E}[\hat{H}(M)] \approx
H_q(M) -
\frac{K(M)-1}{2N_s}.
\end{align*}
For any $M \supseteq \Mstar$, minimality of $\Mstar$ implies $H_q(M) = H_q(\varepsilon)$. Additional variables do not reduce true conditional entropy. However, larger masks have more cells and incur larger bias. The expected apparent advantage of $\ME \supsetneq \Mstar$ over $\Mstar$ is:
\begin{align*}
\mu_\delta \approx
\frac{K(\ME) - K(\Mstar)}{2N_s}
> 0,
\end{align*}
so the entropy objective systematically favors larger masks in finite samples.

\noindent\textbf{Step 2: Estimation variance.}
The per-cell entropy $\hat{H}(B_k)$
(eq.~\eqref{eq:entropy}) is
estimated from $n_k = N_s/K(M)$
observations. Standard concentration results for plug-in entropy estimators \cite{cover2006elements} imply that its variance scales inversely with the number of samples per cell and is proportional to the magnitude of the underlying entropy. Thus, $\mathrm{Var}(\hat{H}(B_k)) = O\bigl(H_q(\varepsilon)^2 \cdot K(M)/N_s\bigr)$. The standard deviation of the entropy difference between two masks is therefore:
\begin{align*}
\sigma_\delta \approx
\frac{H_q(\varepsilon)
\sqrt{K(M_{\max})}}{\sqrt{N_s}},
\end{align*}
where $K(M_{\max})$ is the worst-case bound.

\noindent\textbf{Step 3: Condition for correct selection.}
Correct selection requires the genuine entropy signal to overcome the bias $\mu_\delta$ by at least $z_\alpha$ standard deviations, i.e.\ $\mu_\delta/\sigma_\delta \leq z_\alpha$. Substituting:
\begin{align*}
\frac{K(\ME) - K(\Mstar)}{2N_s}
\cdot
\frac{\sqrt{N_s}}
{H_q(\varepsilon)
\sqrt{K(M_{\max})}}
\leq z_\alpha.
\end{align*}
Solving for $N_s$ and using
$K(\ME) - K(\Mstar) \leq
K(M_{\max})$:
\begin{align*}
N_s \geq C_\alpha\,
\frac{K(M_{\max})}
{H_q(\varepsilon)^2},
\quad C_\alpha = 4z_\alpha^2.
\end{align*}
The constant $C_\alpha$ is not tight. The bound captures the correct scaling in $K(M_{\max})$, $H_q(\varepsilon)$, and $\alpha$. Since $H_q(\varepsilon) \leq \log_2 q$ (eq.~\eqref{eq:hq_ceiling}), substituting gives the conservative bound:
\begin{align*}
N_s \gtrsim
\frac{C_\alpha\, K(M_{\max})}
{(\log_2 q)^2}
\quad \Longleftrightarrow \quad
\lambda \gtrsim
\frac{C_\alpha}{(\log_2 q)^2}.
\end{align*}
\end{proof}

\begin{remark}
When $\lambda < 1$, the discretization ceiling $H_q(\varepsilon) \leq \log_2 q$ prevents the entropy signal from overcoming estimation bias unless $N_s$ scales with $K(M_{\max})$.
\end{remark}

Increasing $q$ reduces $\rho$ but increases $K(M)$, reducing $\lambda$; decreasing $q$ does the opposite. Theorem~\ref{thm:threshold} governs the ceiling via $\rho$ and Theorem~\ref{thm:sample} governs sparsity via $\lambda$. A theoretically guided choice satisfies both conditions simultaneously: $\rho < 1$ requires $q < 1/\sigma_\varepsilon$, while $\lambda \geq C_\alpha/(\log_2 q)^2$ bounds the sample  requirement. The largest $q$ satisfying the resolved regime condition is $q^\star = \lfloor 1/\sigma_\varepsilon \rfloor$; since $\sigma_\varepsilon$ is unknown in practice, $q \in \{2, 3\}$ with regime diagnosis via Table~\ref{tab:principle} provides a reliable guideline. Data-adaptive $q$ selection is left for future work.


\begin{remark}[Non-Gaussian stochasticity]
The Gaussian assumption is used only to derive the $\rho=1$ boundary in Theorem~1. The discretization ceiling $H_q(M) \leq \log_2 q$ holds for any output distribution \cite{cover2006elements}, and the rule table estimates $\hat{P}(Y \mid B_k)$ directly from cell counts without any distributional assumption. Non-Gaussian stochasticity shifts the regime boundary but preserves the resolved/saturation structure.
\end{remark}

\subsection{Deterministic Regime and Mask Faithfulness}
\label{subsec:prop}
The following two results establish when entropy masking uniquely identifies $\Mstar$ and when MSE masking cannot.

\begin{proposition} \label{prop:deterministic}
Suppose no intrinsic stochasticity exists, so $Y_t$ is a deterministic function of $X_t^{M^\star}$.
\begin{enumerate}
\item \textbf{Entropy uniqueness}: $H_q(\Mstar) = 0$, and $\Mstar$ is the unique  minimizer of $H_q(M)$.
\item \textbf{MSE non-uniqueness}: Any mask $M$ achieving perfect prediction satisfies $R^2(M)=0$, and such masks need not coincide with $\Mstar$. 
\end{enumerate}
\end{proposition}
\begin{proof}
\textbf{Part 1.} Since $Y$ is a deterministic function of $X_t^{M^\star}$,
all observations within a cell $B_k$ under $M^\star$ share the same
output value, so $H_q(\Mstar) = 0$. If $M \not\supseteq \Mstar$, at least one relevant variable is omitted; observations in the same cell $B_k$ can have different outputs (since the omitted variable is not controlled for), so $H_q(M) > 0$. Minimality follows because any proper subset of $\Mstar$ fails to achieve zero entropy.

\textbf{Part 2.} Any mask that partitions the data such that all observations within each cell share the same output value achieves zero prediction error, since the cell mean $\hat{\mu}_{B_k}$ equals $Y$ for all observations assigned to that cell. Non-uniqueness arises because redundant or correlated inputs can reproduce the same partition. For example, in a periodic system, a longer lag may perfectly predict $Y$ without being causally minimal.
\end{proof}

\begin{table*}[ht]
\centering
\caption{Regime-based mask selection rule.
$\dagger$: faithfulness required for exact recovery.
$\ddagger$: empty cells default to global mean.}
\label{tab:principle}
\renewcommand{\arraystretch}{1.25}
\small
\begin{tabular}{llllll}
\hline
\textbf{Regime} & \textbf{Condition} &
\textbf{Structural Recovery ($\ME$)} &
\textbf{Prediction ($\MR$)} &
\textbf{Diagnostic} \\
\hline
Deterministic &
$\rho \approx 0$ &
$\ME = \Mstar$ exactly &
$\hat{R}^2$ &
--- \\
Resolved, dense &
$\rho < 1,\; \lambda \geq C_\alpha/H_q(\varepsilon)^2$ &
$\ME = \Mstar$$^\dagger$;\; else $\ME \supseteq \Mstar$ &
$\hat{R}^2$ &
Permutation test \\
Resolved, sparse &
$\rho < 1,\; 1 \leq \lambda < C_\alpha/H_q(\varepsilon)^2$ &
$\ME \supseteq \Mstar$;\; over-selection likely &
$\hat{R}^2$ &
Permutation test \\
Collapse &
$\rho < 1,\; \lambda < 1$ &
Unreliable$^\ddagger$
&
$\hat{R}^2$ &
Check $\lambda$ \\
Ceiling &
$\rho \gtrsim 1$ &
Unreliable (any $\lambda$, $N$) &
$\hat{R}^2$ &
Check $\rho$ \\
\hline
\end{tabular}
\end{table*}

\begin{proposition}
Consider $X_t$ and $Y_t$ as input and outputs with intrinsic stochasticity $\sigma_\varepsilon>0$ and minimal mask $\Mstar=\{X_1\}$. Let $X_2$
be a proxy variable related to $X_1$ through an unknown function $g: \mathbb{R} \to \mathbb{R}$ as $X_1 = g(X_2) + \eta$, where $\eta$ is independent of $X_2$ with $\mathrm{Var}(\eta) = \sigma^2_\eta > 0$.
\begin{enumerate}
\item \textbf{Entropy faithfulness}: If $\rho<1$, then $H_q(\{X_1\}) < H_q(\{X_2\})$. Entropy masking uniquely prefers the true structural mask.
\item \textbf{MSE non-faithfulness}: $R^2(\{X_2\}) = \sigma_\varepsilon^2 + \sigma_\eta^2$, so as $\sigma_\eta^2 \to 0$, $R^2(\{X_2\}) \to R^2(\{X_1\})$: MSE becomes indifferent between the true structural and proxy masks.
\end{enumerate}
\end{proposition}
\begin{proof}
\textbf{Part 1.} Conditioning on $X_1$ leaves only $\varepsilon$ unresolved, so
$H_q(\{X_1\}) = H_q(\varepsilon).$ Conditioning on $X_2$ leaves additional uncertainty due to the
imperfect mapping from $X_2$ to $X_1$. Since $X_1 = g(X_2) + \eta$ with $\sigma_\eta > 0$, observations with the same value of $X_2$ can correspond to different values of $X_1$, and hence to different conditional means of $Y$. As a result, $H_q(\{X_2\}) > H_q(\{X_1\}).$ When $\rho < 1$, discretization preserves this separation: different cells under $X_1$ induce distinguishable output distributions, so the entropy gap remains.

\textbf{Part 2.}
The MSE of the proxy mask decomposes as:
\begin{align*}
R^2(\{X_2\}) = \sigma^2_\varepsilon +
\EE[(\mu_{X_1} - \mu_{X_2})^2]
= \sigma^2_\varepsilon + \sigma^2_\eta,
\end{align*}
where $\sigma^2_\eta$ captures the additional uncertainty from predicting $X_1$ via $X_2$. As $\sigma^2_\eta \to 0$, $R^2(\{X_2\}) \to R^2(\{X_1\}) = \sigma^2_\varepsilon$, the MSE objective cannot distinguish the proxy from the true structural variable.
\end{proof}
Proposition~2 formalizes the \emph{explainability} advantage of entropy masking: when $\rho < 1$, entropy correctly identifies $X_1$ as the true structural driver and rejects the proxy $X_2$, whereas MSE becomes indifferent between them as $\sigma^2_\eta \to 0$. This distinction is practically important when the goal is to understand which variables causally drive the output rather than merely predict it.

\subsection{Regime-Based Decision Principle}
\label{rem:principle}
Table~\ref{tab:principle} summarizes mask selection across the four regimes defined by $\rho$ (\ref{eq:rho}) and $\lambda$ (\ref{eq:lambda}). When \emph{explainability} is the goal, entropy masking in the resolved regime identifies the causally relevant input variables; when prediction alone is required, $\hat{R}^2$ is the robust criterion. The excess prediction risk from substituting $\ME$ for $\MR$ is bounded by $\delta$ (Theorem~\ref{thm:delta}). The permutation test assesses whether each variable's entropy contribution exceeds chance by shuffling input values across observations and comparing the observed entropy reduction against the permuted null distribution. A variable is retained if its contribution is statistically significant ($p < \alpha$), providing practical evidence against over-selection when exact recovery of $M^\star$ cannot be guaranteed, as demonstrated in Section~V.

\section{Experimental Validation}
\label{sec:experiments}

We validate the theoretical results using a two-state symmetric Markov chain:
\begin{align}
  S_{t+1} = \begin{cases}
    1 - S_t & \text{prob } 1{-}\varepsilon \\
    S_t     & \text{prob } \varepsilon
  \end{cases}, \quad S_t \in \{0,1\}.
\end{align}
The true minimal mask is $\Mstar = \{\text{lag-1}\}$, with $\sigma^2_\varepsilon = \varepsilon(1-\varepsilon)$ and binary entropy
$
H_2(\varepsilon)
=
-\varepsilon \log_2 \varepsilon
-
(1-\varepsilon)\log_2(1-\varepsilon).
$
Candidate masks: all $2^5 - 1 = 31$ subsets of lags $\{1,\ldots,5\}$; $q = 3$ bins; $N_s = 3000$;
100 Monte Carlo runs per $\varepsilon \in \{0.01,\ldots,0.40\}$. The entropy estimator includes the Miller-Madow correction~\cite{miller1955note}.

In this case, the stochasticity parameter directly controls the overlap between output distributions. This helps visualize the transition between the regimes predicted by the theory.

Fig.~\ref{fig:validation} shows four validation outcomes. \textbf{(a)}~$\hat{R}^2$ for $\MR$ tracks the variability floor $\varepsilon(1-\varepsilon)$ within 1--2\% at every variability level ($p < 0.001$, paired $t$-test), confirming Theorem~\ref{thm:threshold}(3). \textbf{(b)}~Excess risk $\delta$ grows monotonically from $0.0004$ to $0.012$ and stays below the Theorem~\ref{thm:delta} bound throughout, with the gap widening at high entropy. \textbf{(c)}~The true lag (lag-1) appears in $\ME$ in 100\% of runs at all variability levels, confirming $\ME \supseteq \Mstar$ (Theorem~\ref{thm:sample}). However, exact recovery $\ME = \Mstar$ occurs only at the lowest variability levels ($12\%$ at $\varepsilon = 0.01$), confirming that entropy masking over-selects. $\MR$ exact recovery degrades from 93\% to 11\% as variability increases. \textbf{(d)}~Entropy masking selects all 5 lags at all but the lowest variability level. The sample complexity threshold is crossed between $\varepsilon = 0.10$ and $\varepsilon = 0.15$ (threshold $\approx 2{,}990$ vs.\ $N_s = 2{,}095$), placing the left portion of each panel in the sparse regime and the right above it. Multi-variable coupled scenarios are validated in the grid reliability application (Section~V).

\begin{figure}[t]
\centering
\includegraphics[width=\columnwidth]{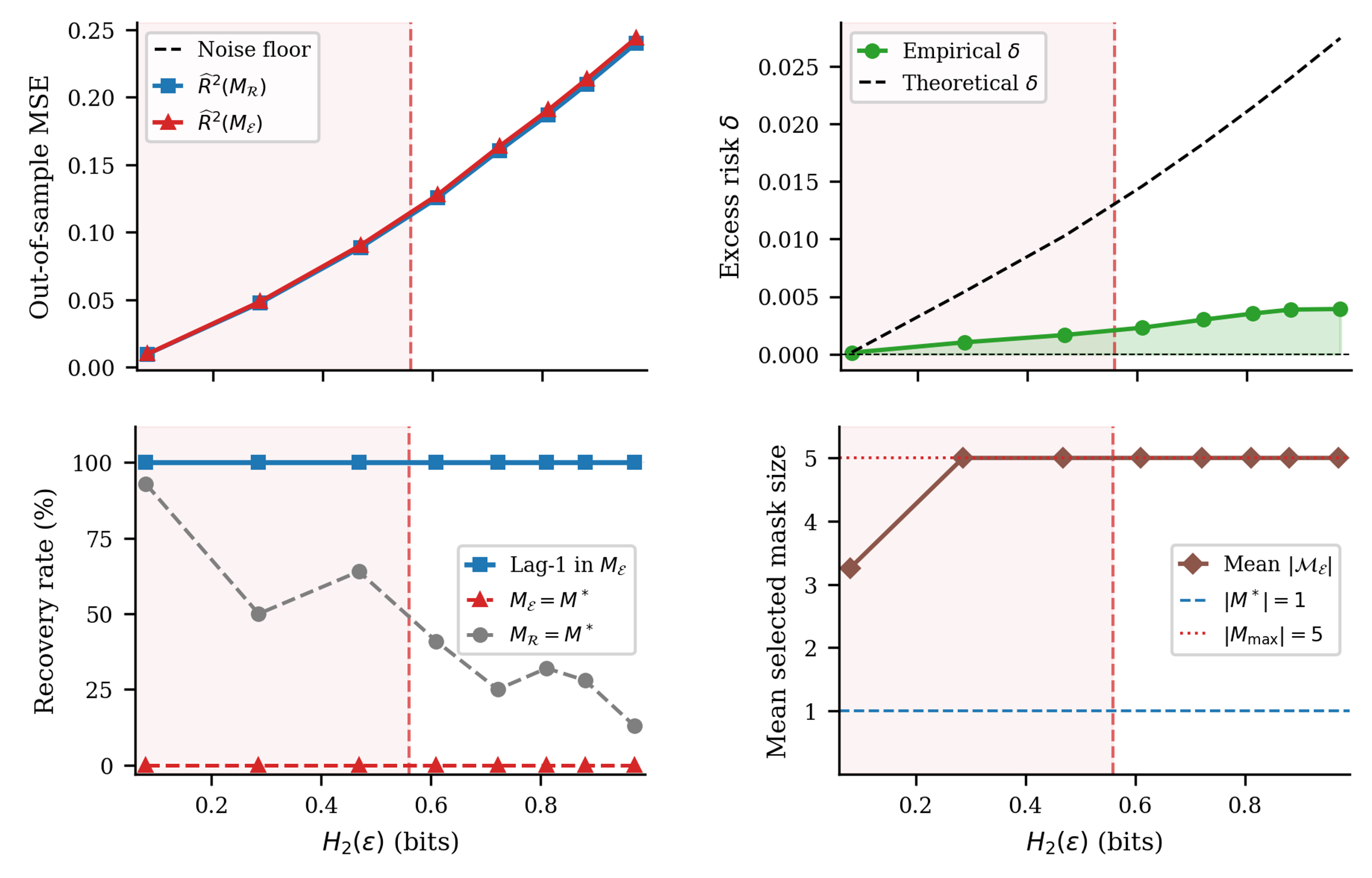}
\caption{Experimental validation on a two-state Markov chain. Red shading marks the sparse regime. (a) Empirical MSE and entropy. (b) Empirical and closed-form $\delta$. (c) Structural ID recovery rate. (d) Mean selected mask size.}
\label{fig:validation}
\end{figure}
\section{Application: Distribution Grid Reliability}
\label{sec:application}
Assessing whether grid modernization improves reliability, and with what delay, remains a key challenge in distribution systems. Advanced Metering Infrastructure (AMI), i.e., smart meters enabling automated outage detection and faster response, has seen widespread deployment, yet its reliability impact remains unclear due to gradual adoption, nonlinear effects, and heterogeneity across utilities. FIR models input--output relationships using discretized states, enabling associative analysis, prediction, and counterfactual simulation. We analyze data from 1300 U.S. distribution utilities (2013–2024) \cite{eia861m_sales_revenue}. The output is $\mathrm{CAIDI}_{i,t}$ (minutes per interruption). Inputs include lagged CAIDI $R_{t-\ell}$, revenue per customer $S_{t-\ell}$, and AMI penetration $A_{t-\ell}$, with lags $\ell \in \{0,1,2,3\}$. Variables are discretized into $q \in \{2,3\}$ levels. Note that this application provides complementary validation on a multi-variable coupled scenario with three input variables ($R$, $S$, $A$) at mixed discretization resolutions ($q_R=2$, $q_S=3$, $q_A=3$), confirming the framework generalizes beyond single-variable settings.

\subsection{Regime Diagnosis}
Before selecting masks, we diagnose which regime the data occupies using Theorems~\ref{thm:threshold}--\ref{thm:sample}. The largest candidate mask has $K(M_{\max}) = 5{,}832$ cells. For input structure identification we have $N_s = 9{,}132$, and so $\lambda \approx 1.57$, i.e., 
the data lies near the sparse boundary. Under Poisson occupancy, approximately, 21\% and 35\%
of cells are empty respectively, so per-cell estimation variance remains high and entropy masking may over-select. The permutation test is therefore essential to screen spurious entropy reductions. 
The relevant diagnostic for structural recovery is $\rho$: with $q = 3$ output bins, $H_q(\varepsilon) \leq \log_2 3 \approx 1.58$ bits and $\rho < 1$, confirming the stochastic resolved regime. The decision principle (Table~\ref{tab:principle}) prescribes: entropy masking for input structure identification, permutation testing for confirmation, and MSE for prediction.

\subsection{Input Structure Mask Analysis}
Minimizing $\hat{H}(M)$ over the full panel yields $\ME = \{A(t{-}2),\; A(t{-}3),\; I(t{-}1),\; R(t{-}1)\}$. The 2--3 year AMI lag is interpretable: deployment takes 2--3 years to reach penetration sufficient for fault-isolation benefits. Since $\rho < 1$ and $\lambda > 1$, Theorem~\ref{thm:sample} guarantees $\ME \supseteq \Mstar$.

AMI values are shuffled across utilities within each year (300 permutations), preserving marginal distributions but breaking the AMI--CAIDI association. This permutation test (300 shuffles, $p = 0.003$, Fig.~\ref{fig:permtest}) confirms a non-random entropy contribution. Counterfactual simulation evaluates $\hat{P}(Y \mid B_k)$ directly from the rule table without refitting. Moving from low to high AMI in the recovery state ($R_{t-1}{=}\text{poor}$, $S_{t-1}{=}\text{low}$) shifts $P(\text{poor}): 0.410{\to}0.389$ and $P(\text{good}): 0.171{\to}0.197$ (Fig.~\ref{fig:counterfactual}),  confirming AMI's protective effect in the recovery state.

\begin{figure}[h]
\centering
\begin{subfigure}[b]{0.48\columnwidth}
    \centering
    \includegraphics[width=\linewidth]{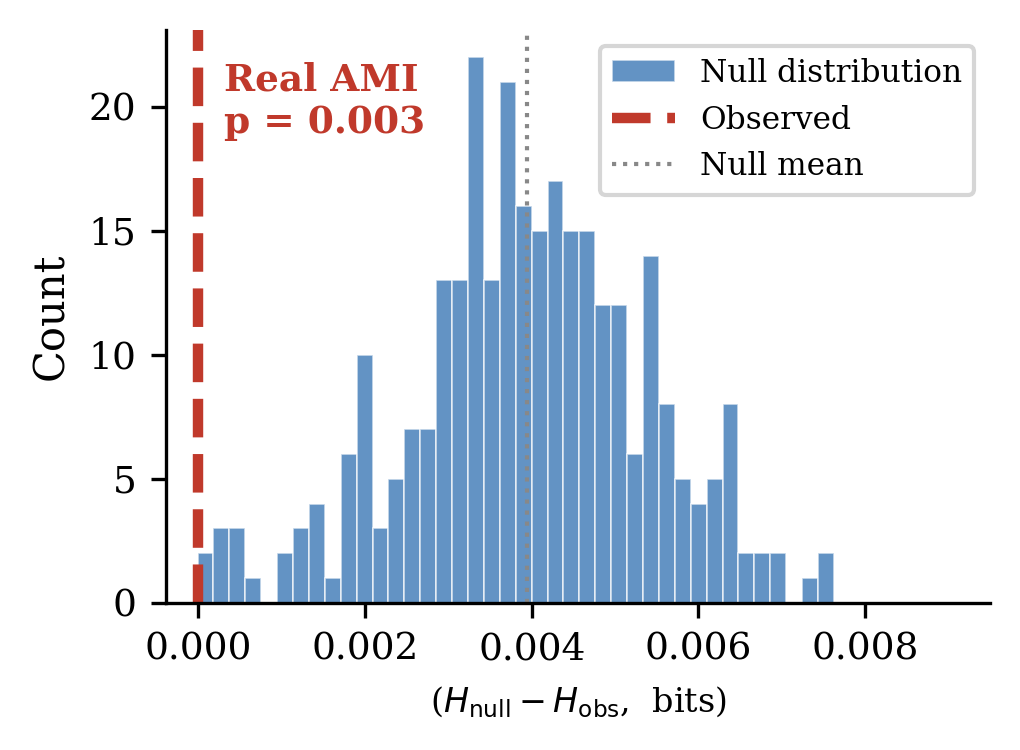}
    \caption{Permutation test ($p=0.003$).}
    \label{fig:permtest}
\end{subfigure}
\hfill
\begin{subfigure}[b]{0.48\columnwidth}
    \centering
    \includegraphics[width=\linewidth]{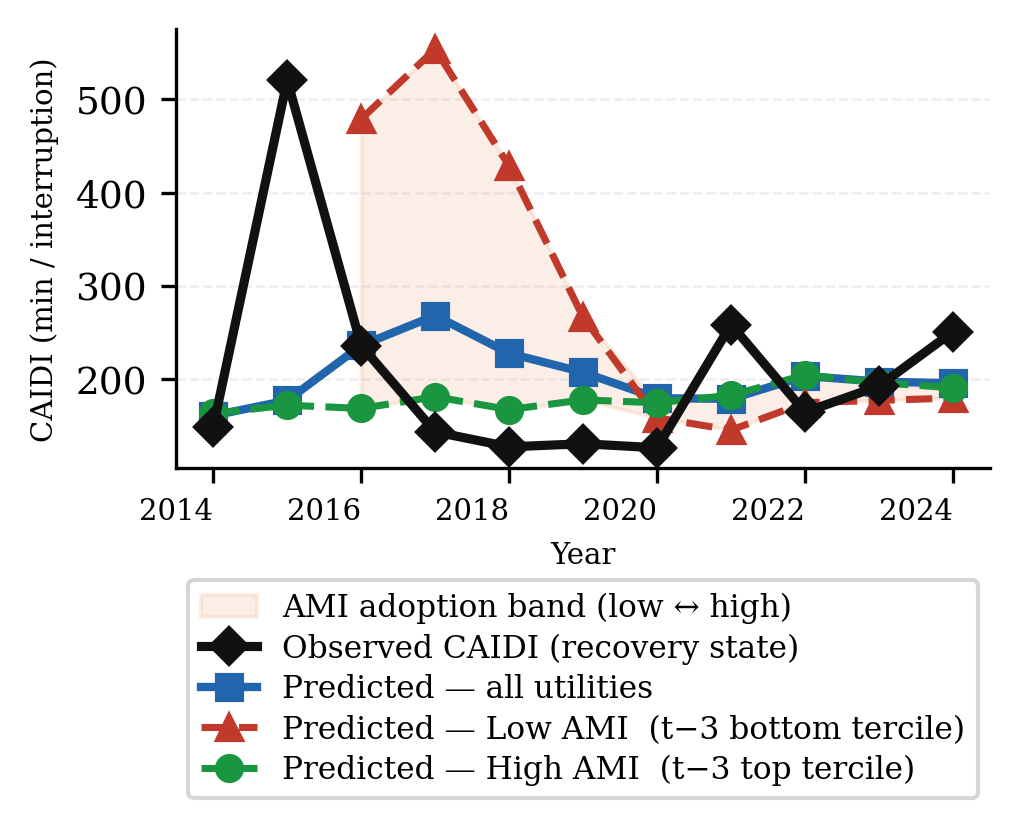}
    \caption{Counterfactual: recovery state.}
    \label{fig:counterfactual}
\end{subfigure}
\caption{Associative evidence for AMI's contribution to CAIDI under the FIR rule-table structure. \textbf{(a)}~Permutation test. \textbf{(b)}~Predicted CAIDI shift under low vs.\ high AMI ($R_{t-1}{=}\text{poor}$, $S_{t-1}{=}\text{low}$).}
\label{fig:causal_evidence}
\end{figure}

\subsection{Prediction}
To evaluate Theorem~\ref{thm:delta}, we consider $N_s = 8427$. The masks are re-optimized on this sample, yielding $\MR = \{R(t),\; R(t{-}1)\}$ with $K(\MR) = 4$ and $\ME = \{A(1,2,3),\; I(1,2,3),\; R(0,1,2)\}$ with $K(\ME) = 5{,}832$. The corresponding MSE values are
$\hat{R}^2(\MR) = 102.38 < \hat{R}^2(\ME) = 183.84$ (Fig.~\ref{fig:rmse_compare}). Since $M^\star$ is unobserved, $\MR$ serves as its empirical proxy; $\delta_{\mathrm{emp}} = 183.84^2 - 102.38^2 \approx 23{,}315$, consistent with the Theorem~\ref{thm:delta} bound of $31{,}806$, confirming $\delta > 0$: substituting $\ME$ for $\MR$ incurs 79\% excess prediction error.

\begin{figure}[t]
\centering
\includegraphics[width=\columnwidth]{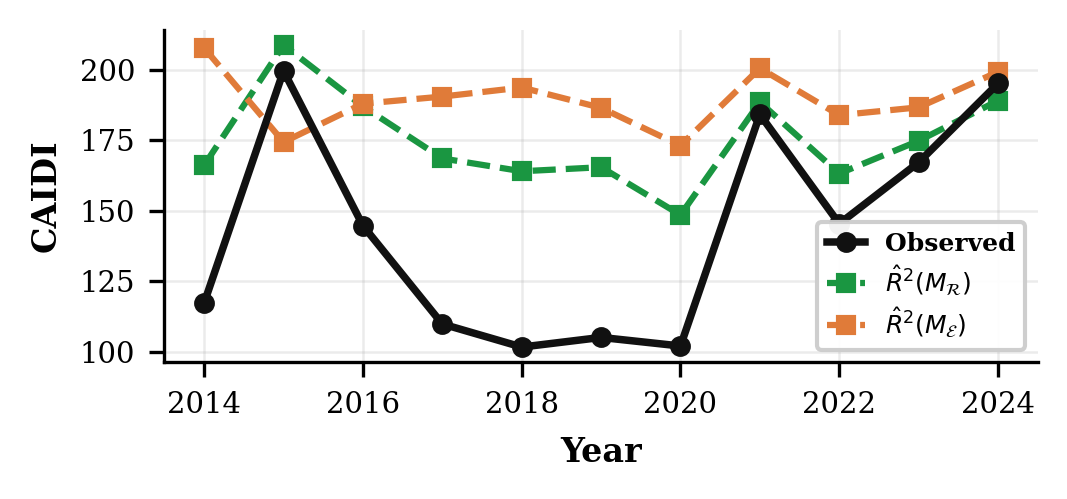}
\caption{Observed vs.\ predicted CAIDI for $\MR$ and $\ME$. 
(Theorem~\ref{thm:delta}).}
\label{fig:rmse_compare}
\end{figure}

\section{Conclusion}
\label{sec:conl}
This work characterizes limits of entropy-based structure identification in discretized nonlinear systems under stochasticity and data sparsity. Three results follow: (1) the resolution-stochasticity ratio $\rho$ governs mask recovery: entropy masking selects an informative mask for $\rho < 1$ but loses discriminative power otherwise; (2) entropy over-selection induces excess prediction risk $\delta$; and (3) exact structural recovery requires data scaling with the number of cells and inversely with the square of the intrinsic coarsened entropy, implying rapidly increasing data requirements as discretization resolution grows. Results are validated on a Markov chain and on distribution grid reliability data. The framework clarifies the distinct roles of the two objectives. Entropy-based selection targets explainability by identifying which variables causally drive the output, while MSE-based selection targets prediction, providing a principled criterion for choosing between them. Two limitations remain. First, the excess-risk bound treats $K(\ME)$ as fixed; when mask size varies, Jensen's inequality~\cite{cover2006elements} implies the bound understates expected excess risk. Second, the threshold $\rho_c$ assumes Gaussian stochasticity, although the regime structure does not depend on the specific distribution. Several directions remain open such as tightening the excess-risk bound via a Jensen correction for random mask size; extending the framework to dependent observations using mixing-time arguments; and developing data-adaptive criteria for selecting $q$ to jointly satisfy the $\rho$ and $\lambda$ conditions.


\balance 
\bibliographystyle{IEEEtran}
\bibliography{Ref}

@article{cellier1996combined,
  title={Combined qualitative/quantitative simulation models of continuous-time processes using fuzzy inductive reasoning techniques},
  author={Cellier, Francois E and Nebot, {\`A}ngela and Mugica, Francisco and Albornoz, Alvaro De},
  journal={International Journal of General System},
  volume={24},
  number={1-2},
  pages={95--116},
  year={1996},
  publisher={Taylor \& Francis}
}

@article{miller1955note,
  title={Note on the bias of information estimates},
  author={Miller, George},
  journal={Information theory in psychology: Problems and methods},
  year={1955},
  publisher={Free Press}
}

@article{tang2024fuzzy,
  title={Fuzzy logic approach for controlling uncertain and nonlinear systems: a comprehensive review of applications and advances},
  author={Tang, Hooi Hung and Ahmad, Nur Syazreen},
  journal={Systems Science \& Control Engineering},
  volume={12},
  number={1},
  pages={2394429},
  year={2024},
  publisher={Taylor \& Francis}
}

@misc{eia861m_sales_revenue,
  author       = {{U.S. Energy Information Administration}},
  title        = {},
  year         = {},
  howpublished = {\url{https://www.eia.gov/electricity/data/eia861m/}},
  note         = {}
}

@article{kalhori2022data,
  title={A data-driven knowledge-based system with reasoning under uncertain evidence for regional long-term hourly load forecasting},
  author={Kalhori, M Rostam Niakan and Emami, I Taheri and Fallahi, F and Tabarzadi, M},
  journal={Applied Energy},
  volume={314},
  pages={118975},
  year={2022},
  publisher={Elsevier}
}

@inproceedings{bagherpour2015hierarchical,
  title={A Hierarchical Perspective to Fuzzy Inductive Reasoning},
  author={Bagherpour, Solmaz and Mugica, Fransisco and Nebot, {\`A}ngela},
  booktitle={2015 IEEE International Conference on Fuzzy Systems (FUZZ-IEEE)},
  pages={1--8},
  year={2015},
  organization={IEEE}
}

@article{jurado2017fuzzy,
  title={Fuzzy inductive reasoning forecasting strategies able to cope with missing data: A smart grid application},
  author={Jurado, Sergio and Nebot, {\`A}ngela and Mugica, Fransisco and Mihaylov, Mihail},
  journal={Applied soft computing},
  volume={51},
  pages={225--238},
  year={2017},
  publisher={Elsevier}
}

@article{nebot2012fuzzy,
  title={Fuzzy Inductive Reasoning: a consolidated approach to data-driven construction of complex dynamical systems},
  author={Nebot, {\`A}ngela and Mugica, Francisco},
  journal={International Journal of General Systems},
  volume={41},
  number={7},
  pages={645--665},
  year={2012},
  publisher={Taylor \& Francis}
}

@article{acosta2007optimization,
  title={Optimization of fuzzy partitions for inductive reasoning using genetic algorithms},
  author={Acosta, J and Nebot, A and Villar, Pedro and Fuertes, Josep M},
  journal={International Journal of Systems Science},
  volume={38},
  number={12},
  pages={991--1011},
  year={2007},
  publisher={Taylor \& Francis}
}

@article{mamlook2009fuzzy,
  title={A fuzzy inference model for short-term load forecasting},
  author={Mamlook, Rustum and Badran, Omar and Abdulhadi, Emad},
  journal={Energy Policy},
  volume={37},
  number={4},
  pages={1239--1248},
  year={2009},
  publisher={Elsevier}
}

@book{cover2006elements,
  title={Elements of information theory (wiley series in telecommunications and signal processing)},
  author={Cover, Thomas M and Thomas, Joy A},
  year={2006},
  publisher={Wiley-interscience}
}

@article{peng2005feature,
  title={Feature selection based on mutual information criteria of max-dependency, max-relevance, and min-redundancy},
  author={Peng, Hanchuan and Long, Fuhui and Ding, Chris},
  journal={IEEE Transactions on pattern analysis and machine intelligence},
  volume={27},
  number={8},
  pages={1226--1238},
  year={2005},
  publisher={IEEE}
}

@article{papaioannou2025role,
  title={The role of mutual information estimator choice in feature selection: An empirical study on mRMR},
  author={Papaioannou, Nikolaos and Myllis, Georgios and Tsimpiris, Alkiviadis and Vrana, Vasiliki},
  journal={Information},
  volume={16},
  number={9},
  pages={724},
  year={2025},
  publisher={MDPI}
}

@inproceedings{covert2023learning,
  title={Learning to maximize mutual information for dynamic feature selection},
  author={Covert, Ian Connick and Qiu, Wei and Lu, Mingyu and Kim, Na Yoon and White, Nathan J and Lee, Su-In},
  booktitle={International Conference on Machine Learning},
  pages={6424--6447},
  year={2023},
  organization={PMLR}
}

@article{mielniczuk2022information,
  title={Information theoretic methods for variable selection—a review},
  author={Mielniczuk, Jan},
  journal={Entropy},
  volume={24},
  number={8},
  pages={1079},
  year={2022},
  publisher={MDPI}
}

\end{document}